\documentclass[pra, superscriptaddress,groupedaddress, a4paper,  twocolumn]{revtex4}
\usepackage{amsmath,amsthm,amsfonts,amssymb,times,graphicx,color}
\usepackage{multirow}
\usepackage{physics}
\usepackage{braket}
\usepackage{subfigure}
\usepackage{mathtools}
\usepackage{color}
\mathtoolsset{showonlyrefs}  

\newtheorem{theorem}{Theorem}[section]

\renewcommand{\hat}{}

\newcommand{\E}{\mathbb{E}}
\renewcommand{\P}{\mathbb{P}}

\newcommand{\qub}[2]{
\begin{pmatrix}
#1 \\
#2
\end{pmatrix}
}

\include{kkmacros}

\begin{document}

\title{Quantum control using quantum memory}

\author{Mathieu Roget}
\email{mathieu.roget@ens-lyon.org}
\affiliation{Aix-Marseille Universit{\'e}, Universit{\'e} de Toulon, CNRS, LIS, Marseille, France and ENS de Lyon, D{\'e}partement d'Informatique}

\author{Basile Herzog}
\email{basile.herzog@gmail.com}
\affiliation{Aix-Marseille Universit{\'e}, Universit{\'e} de Toulon, CNRS, LIS, Marseille, France }
\affiliation{Université de Lorraine, LPCT, Nancy, France}

\author{Giuseppe Di Molfetta}
\email{giuseppe.dimolfetta@lis-lab.fr}
\affiliation{Aix-Marseille Universit{\'e}, Universit{\'e} de Toulon, CNRS, LIS, Marseille, France and Quantum Computing Center, Keio University}

\date{\today}
\begin{abstract}
We propose a new quantum numerical scheme to control the dynamics of a quantum walker in a two dimensional space-time grid. More specifically, we show how, introducing a quantum memory for each of the spatial grid, this result can be achieved simply by acting on the initial state of the whole system, and therefore can be exactly controlled once for all. As example we prove analytically how to encode in the initial state any arbitrary walker's mean trajectory and variance. This  brings  significantly closer the possibility of implementing dynamically interesting physics models on  medium term quantum devices, and introduces a new direction in simulating aspects of quantum field theories (QFTs), notably on curved manifold.
\end{abstract}

\maketitle

Quantum control refers to the ability to steer a dynamical quantum system from an initial to a desired target or outcome, with a desired accuracy. Several theoretical and experimental approaches to model controlled wave packets and their application are very useful to pave the way for future simulation or quantum calculation schemes \cite{dong2010quantum, georgescu2014quantum}. In many of these, the physical system to be controlled is driven by an external potential, which needs to be controlled all along the experience, until the target is achieved. Although in this work we do not claim to offer a general theory of quantum control, we provide a new approach in which the control scheme is encoded once and for all into its initial state. The main protagonist here is not a generic quantum system, but a quantum walks (QW) in discrete time \cite{Aharonov1993aa,grossing88,shakeel2014history}. What may seem like a particular choice, in reality offers great potential, given the recognised versatility of this simple system. In fact, QW are a universal computational model \cite{childs2009universal, lovett2010universal}, that spans a large spectrum of physical and biological phenomena, relevant both for fundamental science and for applications. Applications include search algorithms \cite{Grover1996,Ambainis2004,Portugal2013, guillet2019grover} and graph isomorphism algorithms \cite{berry2011two} to modeling and simulating quantum \cite{arrighi2018quantum, arrighi2019curved, hatifi2019quantum, arrighi2020quantum, di2016quantum} and classical dynamics \cite{di2016discrete, brun2003quantum}.  These models have sparked various theoretical investigations covering areas in mathematics, computer science, quantum information and statistical mechanics and have been defined in any physical dimensions \cite{mackay2002quantum, marquez2017fermion} and over several topologies \cite{acevedo2005quantum, rohde2011multi, aristote2020dynamical}. 
QW appear in multiple variants and can be defined on arbitrary graphs. Essentially, these simple systems have two registers: one for its position on the graph and the other is its internal state, often called coin state. It propagates on the graph, conditioned by its internal state, similarly to the classical case, where at each step we flip a coin to determine the direction of the walker. The essential difference is that in the quantum case, the walker propagates in superposition on the graph in various directions starting from a node. This feature allows the quantum walker to explore the graph quadratically faster a classical one, property that make it very useful to design, e.g., efficient search algorithms. However, we don't know many way to control the quantum walker evolution. For instance we can choose the initial condition and the evolution operator to tune the walker's variance $\sigma(t) = a f(t)$, where $a$ is a real prefactor and $f(t)$ is typically a linear function of $t$. However, once these are fixed at the initial time, both $f$ and $a$ remain the same all along the evolution, unless we don't allow the evolution operator to change in an in-homogeneous way at each time-step, as in \cite{Dimolfetta2014aa, Dimolfetta2013aa}, which may be very costly. How can we control the walker's dynamics at our will without having to change the evolution operator? Would it be possible to control, having only the initial condition, the variance or its average trajectory? In this manuscript we argue that, at the price of introducing a quantum memory, the answer is affirmative. Quantum walks with memory have already been studied and come in several variants \cite{ konno2010limit, li2016generic}. As an example, these modified quantum walks may have extra coins to record the walker’s latest path, as in \cite{mcgettrick2010one,rohde2013quantum}. Here, the idea is to define an additional qubit for each site in the grid, with which the walker interacts throughout the evolution. Surprisingly, we will prove that the initial condition of the whole system, memory + walker, is sufficient to control, e.g., the variance and the mean position of the walker for all times. The interest is double : from one hand we provide a simple distributed quantum computational model to control a single qubit along its dynamics, which will not require us to control and adjust the local update rule at each time step; from a totally different perspective, this simple system may suggest an operational way to model and to unitary discretise curved propagation, as argued in \cite{arrighi2017quantum}. 

The manuscript is organised as follows : In Sec. \ref{model} we will provide the definition of the model with and without memory, in one spatial dimension; then, in Sec. \ref{sec:control}, we will prove analytically and numerically how to control the variance and the mean trajectory of a quantum walker, solely via the initial condition of the whole system. Finally, in Sec. \ref{sec:discussion} we discuss and conclude. 

\section{The model}\label{model}

Formally, the simplest but non trivial QW is defined on a Hilbert space which has position  and  velocity (internal "spin" state) components. The \textit{position} Hilbert space $X$ is the set of states $\ket{x}$ where $x \in \mathbb{Z}_N$, and the \textit{velocity} Hilbert space is $V=\mathbb{C}^2$, for which we may choose some orthonormal basis labeled $\{\ket{v^-}, \ket{v^+}\}$. Denote the QW Hilbert space by $\mathcal{H}$,
\begin{equation} 
\mathcal{H}=X\otimes V.
\end{equation}
 The overall state of the walker at time $t \in \mathbb{N}$ may thus be written 
\begin{equation}
\Psi(t)=\sum_x \psi_x^+(t) \ket{x} \otimes \ket{v^+} + \psi_x^-(t) \ket{x} \otimes \ket{v^-},
\end{equation}
 where the scalar field $\psi_x^+$ (resp. $\psi_x^-$) gives, at every position $x\in \mathbb{Z}_N$, the amplitude of the particle being there and about to move right (resp. left). 
 We can write an amplitude vector at time $t$ and  position $x$ over the ordered basis of the coin space $\{\ket{v^+}, \ket{v^-}\}$,
\begin{equation}
\Psi_x(t) =\begin{pmatrix}\psi_{x}^+(t)\\\psi_{x}^-(t)\end{pmatrix}.
\label{eq:ampvec}
\end{equation} 
 Let $W$ be the evolution of the QW at each time step,
\begin{equation}
\Psi_x(t+1)=W \Psi_x(t).
\label{eq:QWeq}
\end{equation}
$W$ is composed  of a coin operator $\hat{C}$, an arbitrary element of $U(2)$, acting on the velocity space, 
 e.g., 
\[ C =\left(\begin{matrix}
\cos{\theta} & i\sin{\theta} \\
i\sin{\theta} & \cos{\theta} \\
\end{matrix}\right), \]\\
followed by a shift operator $\hat{S}$ 
\begin{equation}
\hat{S}\Psi_x(t) =\begin{pmatrix}\psi_{x-1}^+(t)\\\psi_{x+1}^-(t)\end{pmatrix},
\label{eq:shift}
\end{equation} 
with the overall evolution being
\begin{equation}
W = \hat{S}({\text{Id}_X \otimes\hat{C}}),
\end{equation}
where $\text{Id}_X$ is the identity operator on the position Hilbert space.
\paragraph{Quantum memory}
Now, let us consider that at each site $x$ of the grid we have a supplementary qubit $\ket{m_x} \in  \mathcal{M}^x=\mathbb{C}^2$. This extended Hilbert space, as proved in \cite{shakeel2014history, shakeel2019neighborhood}, may be used as a quantum memory to keep track of the past of the walker. Notice that the size of the Hilbert space now seems to be growing exponentially. However, according to \cite{shakeel2018quantum} this is not going to be an issue for infinite lattices as the Hilbert space can be taken to be countably infinite dimensional, even with the memory qubits included. That is because the interactions are only finite neighborhood and a Hilbert space of finite, unbounded configurations suffices \cite{arrighi2011unitarity}. 

The whole state (QW $+$ quantum memory) lies now in $X \otimes  V \bigotimes_{x=1}^{N} \mathcal{M}^x$. Indeed, one of the main motivation of this history dependent QW is to build a truly self-avoiding walker: One knows that the walker moves towards the left or towards the right according to the internal coin state; thus, in order to avoid sites already visited by the walker, one conditions the coin state on the neighbor memory states, which eventually recorded previous presence of the walker. 

The coin operator $\mathcal{C}$ of the previous section is replaced by a different  operator $Q$, that acts on the joint velocity-memory space : $V \otimes \mathcal{H}_{\mathcal{M}^{x-1}} \otimes \mathcal{H}_{\mathcal{M}^{x+1}}$ - whose basis is  the set of $\ket{v^\alpha \beta_l\delta_r}$, $\alpha = \{+, -\}$, $\ket{\beta_l}$ and $\ket{\delta_r}$ being the memory qubits located respectively at the positions $(x-1)$ and $(x+1)$ - so that the memory qubits adjacent to the position $x$ of the walker are involved:

\[\begin{array}{cccl}
\hat Q = \ket{v^+ 1 1 }\bra{v^+ 00} + 
  \ket{v^- 0 0}\bra{v^- 00} + \\
\ket{v^- 0 1}\bra{v^+ 01} + 
\ket{v^+ 0 1}  \bra{v^- 10} + \\
\ket{v^+ 1 0}\bra{v^+10} + 
\ket{v^- 1 0}\bra{v^-01} + \\
\ket{v^+ 0 0 } \bra{v^+11} + 
\ket{v^- 1 1}\bra{v^-11} \\
\end{array}
\] 

\begin{figure}
\includegraphics[width=7cm]{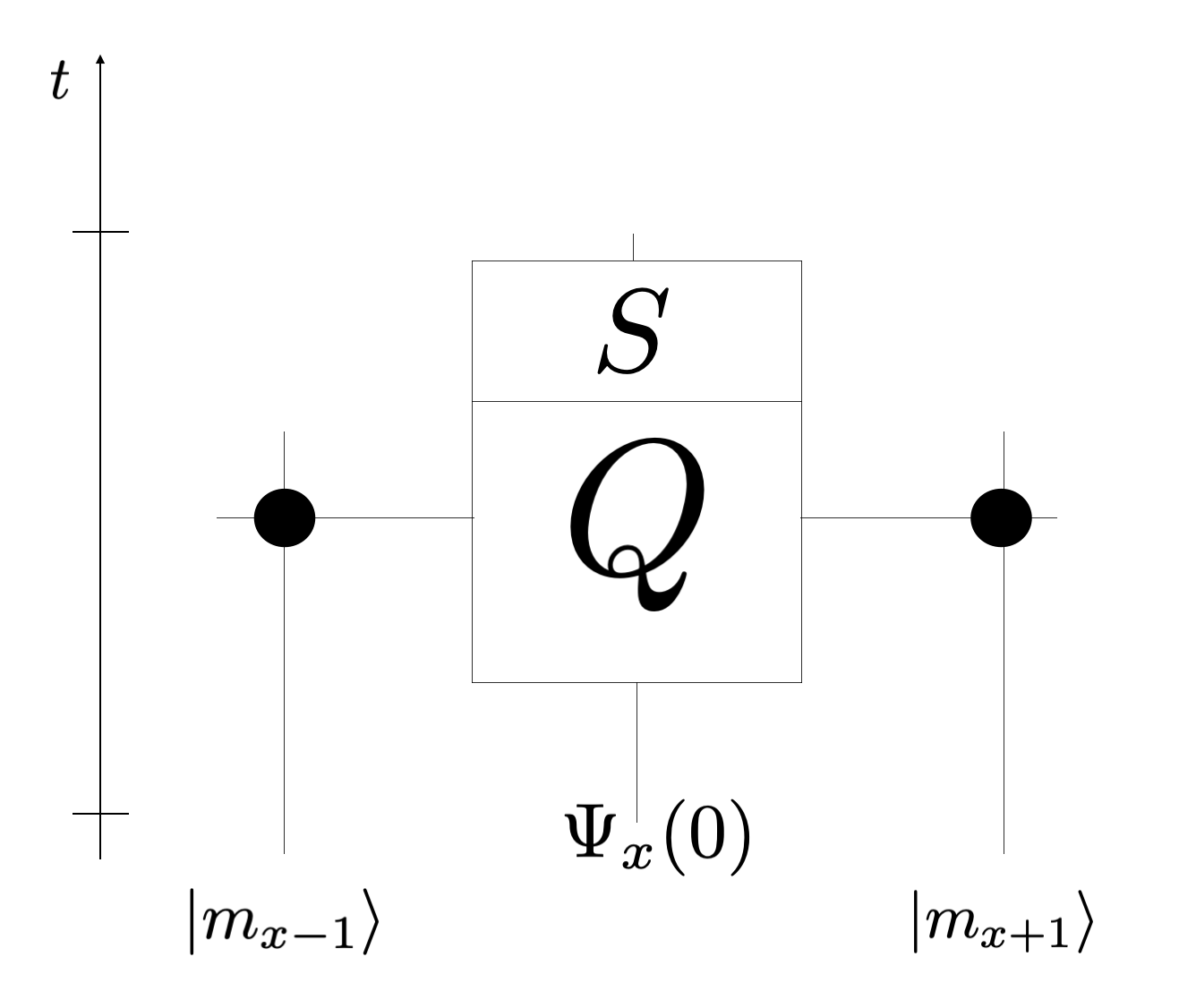}
\caption{Schematic evolution of the walker $\Psi_x(0)$ at the initial time conditioned by the neighbors qubits at position $x=-1$ and $x=1$.} 
\label{fig:evo}
\end{figure}

Finally, the shift operator still acts on the joint velocity-position space in the standard manner, as defined in Eq. \eqref{eq:shift}, and trivially on the memory space. Altogether the global evolution, depicted in Fig. \ref{fig:evo}, is:
\begin{equation}
\Psi(t+1)= G \Psi(t)
\label{eq:evo}
\end{equation}
\begin{equation}
G= S Q .
\end{equation}

\section{Control the walker's dynamics}\label{sec:control}

To simplify our analysis, we  choose a localized set of initial conditions  
\begin{small}
\begin{equation}
\Psi(0) = \ket{0}\ket{v^-} \left(\bigotimes_{x=-\left \lfloor{\frac{N}{2}}\right \rfloor}^{-1}{A_{-x}\ket{0}+B_{-x}\ket{1}}\right)\left(\bigotimes_{x=0}^{\left \lfloor{\frac{N}{2}}\right \rfloor}{\ket{0}}\right)
\label{eq:IC}
\end{equation}
\end{small}
which is equivalent to requiring that all those sites with a nonzero internal state have all quantum memories to their right set to $\ket{0}$, as depicted in Fig. \ref{fig:mem}. 

\begin{figure}
\includegraphics[width=8cm]{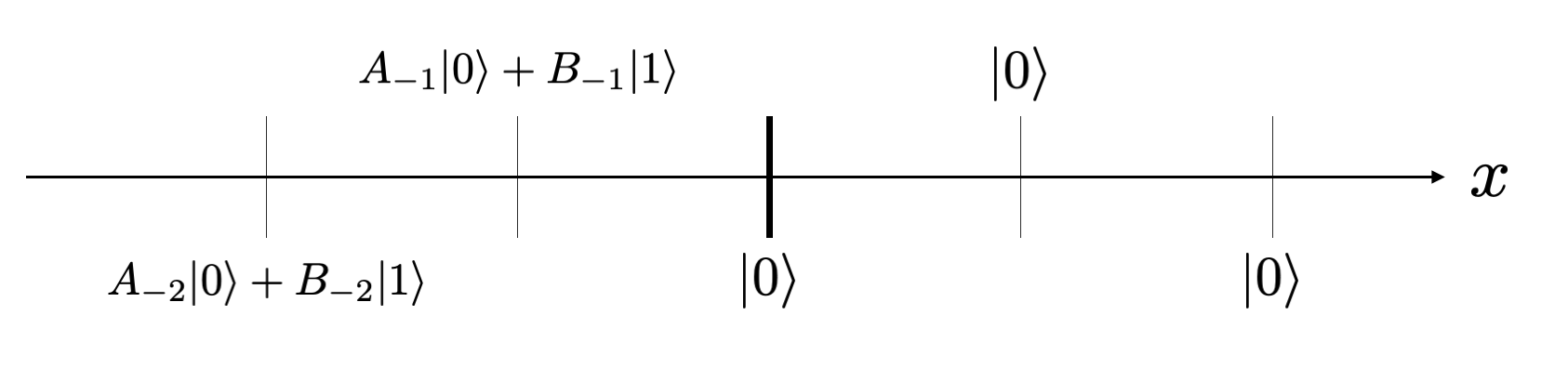}
\caption{The initial state of the quantum memory on the one dimensional line.} 
\label{fig:mem}
\end{figure}

By looking at the single walker subspace, the recurrence relations for each amplitude, detailed in Appendix \ref{ap:ap1}, read : 
\begin{equation}
\begin{split}
\psi_x^-(t+1) &= A_{-x}\psi_{x+1}^-(t) \\ 
\psi_x^+(t+1) &= \psi_{x-1}^+(t)  + B_{-x+2}\psi_{x-1}^-(t)
\label{eq:rec1}
\end{split}
\end{equation}
and given Eq. \eqref{eq:IC} we prove the following theorem :

\begin{theorem}
\label{th:amplitudes}
The left-moving and right-moving amplitudes $\psi_x^{\mp}(t)$, solutions of the linear set of Eqs. \eqref{eq:rec1}, reads:
\begin{equation}
\psi_x^{-}(t)=\left\{\begin{matrix}
\prod_{i=1}^t{A_i} & \text{if } x=-t \\
0 & \text{else} & \\
\end{matrix}\right.
\end{equation}
and 
\[\psi_x^{+}(t)=\left\{\begin{matrix}
0 & \text{if } |x|>t  \\
0 & \text{if } x=-t  \\
0 & \text{if } x-t \text{ is odd} \\
B_{\frac{t-x}{2}+1}\prod_{j=1}^{\frac{t-x}{2}}{A_j} & \text{else} & \\
\end{matrix}\right.\]
\end{theorem}

\begin{proof}
We can prove the above theorem by induction. Indeed, let us first compute $\psi^-_x(t+1)$, we may discern two cases: 
\begin{itemize}
\item $x \neq -(t+1)$ : then $x+1 \neq -t$ and 
\[\psi_{x+1}^-(t) = 0 \Rightarrow \psi^-_x(t+1) = 0\]
\item $x = -(t+1)$ : then $x+1 = -t$ and
\[\psi^- _x(t+1) = A_{-x}\psi^-_{x+1}(t) = A_{t+1}\prod_{i=1}^t{A_i} = \prod_{i=1}^{t+1}{A_i}.\]
\end{itemize}

In order to compute $\psi_x^{+}(t+1)$, we discern four cases : 
\begin{itemize}
\item $|x|>t+1$ : the celerity of the walk being of one space-step per time-step, and the initial condition being localized, then $\psi^+_x(t+1) = 0$
\item $x = -(t+1)$ : using the same argument, $\psi^{\pm}_{x-1}(t+1)=0$ and $\psi^+_x(t+1) = 0$
\item $x-(t+1)$ is odd : then $(x-1)-t$ is odd and $x-1 \neq -t$,
\[\psi^+_{x-1}(t) = 0 \text{ and } \psi^-_{x-1}(t) = 0 \Rightarrow \psi^+_x(t+1) = 0\]
\item else (in particular, $x-t-1$ is even) : there are two subcases :
\begin{itemize}
\item[$\blacktriangleright$] $x-1 = -t$ : then $\psi^+_{x-1}(t) = 0$ and
\[\psi^+_x(t+1) = B_{-x+2}\prod_{i=1}^t{A_i} = B_{\frac{t+1-x}{2}+1}\prod_{j=1}^{\frac{t+1-x}{2}}{A_j}\]
\item[$\blacktriangleright$] $x-1 \neq -t$ : then $\psi^-_{x-1}(t) = 0$ and 
\[\psi^+_x(t+1) = B_{\frac{t+1-x}{2}+1}\prod_{j=1}^{\frac{t+1-x}{2}}{A_j}\]
\end{itemize}
\end{itemize}
\end{proof}

Using Theorem \ref{th:amplitudes}, it is straightforward to compute the probability density of the walker, which reads:

\[\P_x(t)=\left\{\begin{matrix}
\prod_{i=1}^t{A_i^2} & \text{if } x=-t \\
0 & \text{if } x-t \text{ is odd} \\
B_{\frac{t-x}{2}+1}^2\prod_{j=1}^{\frac{t-x}{2}}{A_j^2} & \text{else} & \\
\end{matrix}\right.\]
where, the above probability density is vanishing for $|x|>t$.\\
Once we know the analytical expression of the probability density, one can compute the mean trajectory and variance. The first one reads :

\[\E(t) = -t\prod_{i=1}^t{A_i^2}+ \sum_{k=0}^{t-1}{(t-2k)B_{k+1}^2\prod_{j=1}^{k}{A_j^2}},\]

and the variance:

\begin{small} 
 \begin{multline}
\text{Var}(t)=t^{2} \prod_{i=1}^{t} A_i^{2} -\\ \nonumber 
 \left(t \prod_{i=1}^{t} A_i^{2} - \sum_{i=0}^{t - 1} - \left(2 i - t\right) B_{i+1}^{2} \prod_{j=1}^{i} A_j^{2}\right)^{2} + \\ \nonumber
\sum_{i=0}^{t - 1} \left(2 i - t\right)^{2} B_{i+1}^{2}\prod_{j=1}^{i} A_j^{2}.
\end{multline}
\end{small}

Notice that the probability density and both the first two momenta, depend solely on the parameters $\{A_x,B_x\}$ which we fix at the initial state.
Let us now explore few exemples to show how we can recover any arbitrary mean trajectory and variance by controlling the sole initial condition of the system. Suppose we need to recover a \textit{linear} mean trajectory. The way to do that is setting $A_k=1$ and $B_k=0$ $\forall k > 1$, supposed that $A_1$ and $B_1$ are known. Without lack of generality we set $A_1^2 = (1-B_1^2)$.\\
Then the probability density is:
\[\P_x(t)=\left\{\begin{matrix}
(1-B_1^2)^2 & \text{if } x=-t \\
B_1^2 & \text{if } x=t \\
0 & \text{else} \\
\end{matrix}\right.\]
In this particular case, the general expression of the mean value of the trajectory reduces to:
\begin{equation}
\E(t) = - t (1-2B^{2}_1)
\end{equation}
The pre-factor $2B_1^2-1 \equiv v$, may be seen as the mean velocity of the walker. As we can see in Fig. \ref{fig:VElinear} the smaller is $v$, the smaller will be the velocity of the mean trajectory. The variance $\text{Var}(t)$  will be of course $\propto t^2$, which coincide with the standard ballistic behaviour of an homogeneous QW. 

\begin{figure}
\includegraphics[width=8cm]{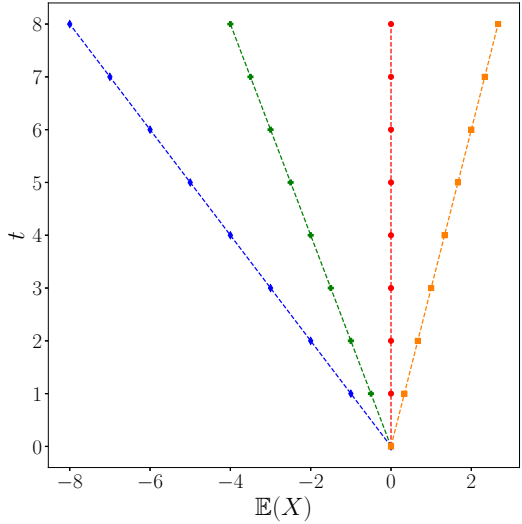}
\caption{Mean trajectory of the walker for different values of $v$. 
Points represents the theoretical prediction, the dashed line coincides with the numerical simulation.}
\label{fig:VElinear}
\end{figure}

A second less trivial example may be represented by the \textit{parabolic} mean trajectory, which translates in a non-linear variance. For example, let's set $\displaystyle B_k = \sqrt{- \frac{z}{z \left(k - 1\right) - 2}}$, $z \in \mathbb{R}^+$. \\ Now the probability reads :
\[\P_x(t)=\left\{\begin{matrix}
1-\frac{tz}{z+2} & \text{if } x=-t \\
0 & \text{if } x-t \text{ is odd} \\
\frac{z}{z+2} & \text{else} \\
\end{matrix}\right.\]
\begin{figure}[h!]
\includegraphics[width=8cm]{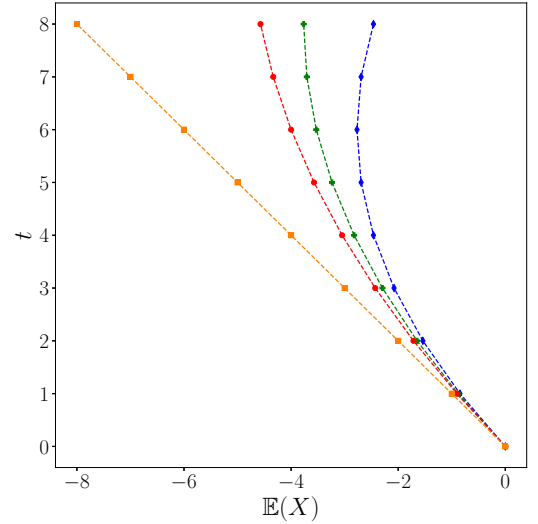}\\
\includegraphics[width=8cm]{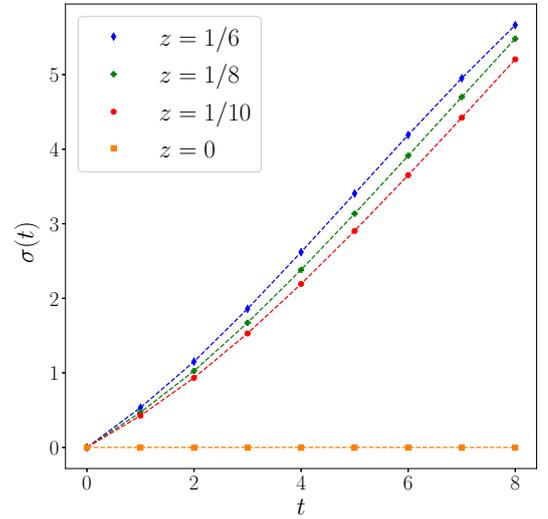}
\caption{Mean trajectory of the walker for different values of $z$. (Bottom) Variance of the walker for different values of $z$. 
Points represents the theoretical prediction, the dashed line coincides with the numerical simulation.}
\label{fig:varNL}
\end{figure}
and we can consequently deduce the mean trajectory
\[\E(t) = \frac{t(zt-2)}{z+2}\]
and the standard deviation $\sigma =\sqrt{\text{Var}(t)} $ :
\begin{small}
\[\sigma(t) = \sqrt{\frac{t \left(- 3 t \left(t z - 2\right)^{2} + \left(z + 2\right) \left(- 2 t^{2} z + 3 t z + 6 t + 2 z\right)\right)}{3 \left(z + 2\right)^{2}}}.\]
\end{small}
Notice that the above standard deviation is not linear in time, as shown in Fig. \ref{fig:varNL} and the non-linear behaviour of $\sigma$ strongly depends on $z$.

In all previous particular cases, we have shown how to recover linear and non linear moments, keeping the probability density $\P_x(t)$ either constant either linear in $t$. In our last result, we show how it is possible to generalize these results, making the walker's probability density follow arbitrary trajectories. This result is surprising as, although similar results have been obtained before, they required to define a local metrics for each point of the space-time lattice. More specifically, we want to show how, by paying the price of introducing quantum memory, this result can be achieved simply by acting on the initial state and therefore can be exactly controlled once for all. This translates in the following theorem: 

\begin{theorem}
Let us choose \[
  B_k^2 = \frac{f_{k-1}}{1-\sum_{i=1}^{k-2}{f_i}}
 \]
 then the probability density reads:
\[\P_x(t)=\left\{\begin{matrix}
1-\sum_{i=0}^{t-1}{f_i} & \text{if } x=-t\\
0 & \text{if } x-t \text{ is odd} \text{ or } |x|>t\\
f_{\frac{t-x}{2}} & \text{else}\\
\end{matrix}\right.\]
for some arbitrary function $f_i$ which verifies
\begin{itemize}
 \item $\forall t, \; \sum_{i=0}^{t-1}{f_i} \leq 1$
 \item $\forall t, \; 0 \leq f_t \leq 1.$
\end{itemize}

 \end{theorem}
\begin{proof}
 We take  \[
 \left\{
 \begin{matrix}
  B_k^2 = \frac{f_{k-1}}{1-\sum_{i=0}^{k-2}{f_i}}\\ \nonumber
    \\
  A_k^2 = 1-\frac{f_{k-1}}{1-\sum_{i=0}^{k-2}{f_i}}\;\\ \nonumber
 \end{matrix}
\right. . \]
 Now, we want to prove that 
 \[
  \left\{
  \begin{matrix}
   \prod_{i=1}^t{A_i^2} = 1-\sum_{i=0}^{t-1}{f_i}\\
    \\
   B_{\frac{t-x}{2}+1}^2\prod_{j=1}^{\frac{t-x}{2}}{A_j^2} = f_{\frac{t-x}{2}}\\
  \end{matrix}
    \right. .
 \]
First let us prove by induction that \[\prod_{i=1}^t{A_i^2} = 1-\sum_{i=0}^{t-1}{f_i} .\]
\begin{itemize}
 \item For $t=0$ : \[\prod_{i=1}^0{A_i^2} = 1-\sum_{i=0}^{-1}{f_i}  = 1 .\]
 \item If it's true for $t$, then
 \[
  \begin{split}
   \prod_{i=1}^{t+1}{A_i^2} &= A_{t+1}^2\prod_{i=1}^{t}{A_i^2}\\
   &= \left(1-\frac{f_t}{1-\sum_{i=0}^{t-1}{f_i}} \right)\prod_{i=1}^{t}{A_i^2}\\
   &= \left(1-\frac{f_t}{1-\sum_{i=0}^{t-1}{f_i}} \right)\left(1-\sum_{i=0}^{t-1}{f_i}\right)\\
   &= 1-\sum_{i=0}^{t-1}{f_i} - f_t\\
   & = 1 - \sum_{i=0}^{t}{f_i}
  \end{split}.
 \]
 So it is true for $t+1$.\\
 \end{itemize}
 Now let us prove that \[B_{\frac{t-x}{2}+1}^2\prod_{j=1}^{\frac{t-x}{2}}{A_j^2} = f_{\frac{t-x}{2}} \]
 \[
  \begin{split}
   B_{\frac{t-x}{2}+1}^2\prod_{j=1}^{\frac{t-x}{2}}{A_j^2} &= B_{\frac{t-x}{2}+1}^2 \left(1-\sum_{i=0}^{\frac{t-x}{2}-1}{f_i}\right)\\
   &= \left( \frac{f_{\frac{t-x}{2}}}{1-\sum_{i=0}^{\frac{t-x}{2}-1}{f_i}}\right)\left(1-\sum_{i=0}^{\frac{t-x}{2}-1}{f_i}\right)\\
   &= f_{\frac{t-x}{2}}\\
  \end{split}
 \]
\end{proof}

\section{Discussion}\label{sec:discussion}
In conclusion, we introduced a quantum walker which interacts with a memory at each site, allowing the walker's dynamics to depend on the state of the memory qubits in the particle's neighborhood. We considered a scheme that parametrizes the initial states of the memory qubits, and we analytically obtain the probability density of the walker's position, and consequently, its mean trajectory and variance. Varying these parameters alone suffices to generate a range of trajectories that may simulate motion on curved manifolds. This pave the way to implement dynamically interesting physics models, especially quantum particle propagation on curved spacetime. Indeed, embedding the mean trajectory, once for all, in the initial state of the overall QCA, is a clear computational advantage, which may reduce the resources needed for the simulation of a wide variety of dynamical physical models. Possible extensions of these results may also concern quantum algorithms. For example, such a model could inspire efficient spatial search algorithms, as an extension of single quantum walk based schemes. We also leave to future research the extension of the model in higher dimensional space than one. 

\section{Acknowledgements} The authors acknowledge inspiring conversations with Asif Shakeel, David Meyer, Pablo Arrighi and Nathanael Eon. This work has been funded by the Pépinière d’Excellence 2018, AMIDEX fondation, project DiTiQuS and the ID 60609 grant from the John Templeton Foundation, as part of the “The Quantum Information Structure of Spacetime (QISS)” Project.

\appendix

\section{Appendix\label{ap:ap1}}

\subsection{Recurrence relations}

Let's recall the initial state \eqref{eq:IC} : 
\begin{small}
\begin{equation*}
\Psi(0) = \ket{0}\ket{v^-} \left(\bigotimes_{x=-\left \lfloor{\frac{N}{2}}\right \rfloor}^{-1}{A_{-x}\ket{0}+B_{-x}\ket{1}}\right)\left(\bigotimes_{x=0}^{\left \lfloor{\frac{N}{2}}\right \rfloor}{\ket{0}}\right)
\end{equation*}
\end{small}

and let's rephrase it in the following array : \\

$
\begin{array}{||c||c|c|c|c|c|c|c||}
    \hline
    \text{position} & .. & -2 & -1 & 0 & 1 & 2 & ..  \\
    \hline
    \text{velocity} & \qub{..}{..} & \qub{0}{0} & \qub{0}{0} & \qub{1}{0} & \qub{0}{0} & \qub{0}{0} & \qub{..}{..} \\
    \hline
    \text{memory} & \qub{..}{..} & \qub{A_2}{B_2} & \qub{A_1}{B_1} & \qub{1}{0} & \qub{1}{0} & \qub{1}{0} & \qub{..}{..} \\
    \hline
\end{array}
$

\vspace{0.5 cm}

\paragraph{First action of Q} : 
$Q$ can only act on $x=0$, on which there is $\ket{v^-}$ associated to its left and right neighbors which are in memory space $A_1 \ket{0} + B_1 \ket{1}$ (left neighbor in position $x=-1$) and $1 \times \ket{0}$ (right neighbor in position $x=+1$). \\
So that $Q$ will act on $A_1\ket{v^-00}+ B_1\ket{v^-10}$ and
$$ Q \left( A_1\ket{v^-00}+ B_1\ket{v^-10} \right) = A_1\ket{v^-00} + B_1 \ket{v^+01} $$

\paragraph{First shift} : 
The shift acts in the position velocity space on $A_1 \ket{0} \ket{v^-} + B_1 \ket{0} \ket{v^+}$, and its action is  
$$S \left( A_1 \ket{0} \ket{v^-} + B_1 \ket{0} \ket{v^+} \right) = A_1 \ket{-1} \ket{v^-} + B_1 \ket{+1}\ket{v^+} $$

\paragraph{State after first iteration} : \\

$
\begin{array}{||c||c|c|c|c|c|c|c||}
    \hline
    \text{position} & .. & -2 & -1 & 0 & 1 & 2 & ..  \\
    \hline
    \text{velocity} & \qub{..}{..} & \qub{0}{0} & \qub{A_1}{0} & \qub{0}{0} & \qub{0}{B_1} & \qub{0}{0} & \qub{..}{..} \\
    \hline
    \text{memory} & \qub{..}{..} & \qub{A_2}{B_2} & \qub{1}{0} & \qub{1}{0} & \qub{1}{1} & \qub{1}{0} & \qub{..}{..} \\
    \hline
\end{array}$
\vspace{0.5 cm}

In position velocity space : \\
$$\psi(1) = A_1\ket{-1}\ket{v^-}+B_1\ket{+1}\ket{v^+}$$
So that 
$$\psi^-_{-1}(1) = A_1\times 1 = A_{-x}\psi^-_{x+1}(0)$$ and 
$$\psi^+_{1} = B_1 \times 1 = B_{-x+2}\psi^-_{x-1}(0)$$

From there, the recurrence relations for the amplitudes $\psi^{\pm}_x(t)$ reads : 
\begin{align*}
\psi_x^-(t+1) &= A_{-x}\psi_{x+1}^-(t)\\ 
\psi_x^+(t+1) &= B_{-x+2}\psi_{x-1}^-(t)
\label{eq:rec1}
\end{align*}

\paragraph{Second step} : \\
Now $Q$ acts on $x=-1$ and $x=1$ where, \\
for $A_1\ket{x=-1}\ket{v^-}$, we have for the left neighbor $A_2 \ket{0}+B_2\ket{1}$ and for the right neighbor $1 \times \ket{0}$, so that $Q$ will act there on $A_2 A_1 \ket{v^-00}+ B_2 A_1 \ket{v^-10}$, resulting in
$$Q \left( A_2A_1\ket{v^-00} + B_2A_1 \ket{v^-10} \right) = A_2A_1\ket{v^-00} + B_2A_1 \ket{v^+01}$$ 

For $B_1\ket{x=+1}\ket{v^+}$, we have for the left neighbor $1 \times \ket{0}$ and for the right neighbor $1 \times \ket{0}$, so that $Q$ will act on $B_1 \ket{v^+00}$, and 
$$Q B_1 \ket{v^+00} = B_1 \ket{v^+11}$$

Now the shift will act, in the position velocity space, as
\begin{multline*}
 S \left( A_2A_1\ket{-1}\ket{v^-} + B_2A_1\ket{-1}\ket{v^+} + B_1\ket{+1}\ket{v^+} \right) \\
 = A_2A_1\ket{-2}\ket{v^-}+B_2A_1\ket{0}\ket{v^+}+B_1\ket{+2}\ket{v^+}
\end{multline*}



For the amplitudes $\psi^{\pm}$, this gives : 
$$\psi_{-2}^-(2) = A_2\times A_1 = A_{-x}\psi^-_{x+1}(1)$$
 
$$\psi_0^+(2) = B_2 \times A_1 = B_{-x+2}\psi^-_{x-1}(1)$$ 

$$\psi_2^+(2) = B_1 = \psi^+_{x-1}(1)$$

So that now the full recurrence relations are 

\begin{align*}
\psi_x^-(t+1) &= A_{-x}\psi_{x+1}^-(t) \hspace{2cm}\\ \nonumber
\psi_x^+(t+1) &= \psi_{x-1}^+(t)  + B_{-x+2}\psi_{x-1}^-(t)
\end{align*}

\bibliographystyle{plain}	
\bibliography{biblio.bib}

\newcommand{\noop}[1]{}
\begin{thebibliography}{10}

\bibitem{acevedo2005quantum}
O~Lopez Acevedo and Thierry Gobron.
\newblock Quantum walks on cayley graphs.
\newblock {\em Journal of Physics A: Mathematical and General}, 39(3):585,
  2005.

\bibitem{Aharonov1993aa}
Y.~Aharonov, L.~Davidovich, and N.~Zagury.
\newblock Quantum random walks.
\newblock {\em Phys. Rev. A}, 48:1687, 1993.

\bibitem{Ambainis2004}
A.~Ambainis.
\newblock Quantum search algorithms.
\newblock {\em SIGACT News}, 35(2):22--35, June 2004.

\bibitem{aristote2020dynamical}
Quentin Aristote, Nathana{\"e}l Eon, and Giuseppe Di~Molfetta.
\newblock Dynamical triangulation induced by quantum walk.
\newblock {\em Symmetry}, 12(1):128, 2020.

\bibitem{arrighi2020quantum}
Pablo Arrighi, C{\'e}dric B{\'e}ny, and Terry Farrelly.
\newblock A quantum cellular automaton for one-dimensional qed.
\newblock {\em Quantum Information Processing}, 19(3):88, 2020.

\bibitem{arrighi2017quantum}
Pablo Arrighi, Giuseppe Di~Molfetta, and Stefano Facchini.
\newblock Quantum walking in curved spacetime: discrete metric.
\newblock {\em arXiv preprint arXiv:1711.04662}, 2017.

\bibitem{arrighi2018quantum}
Pablo Arrighi, Giuseppe Di~Molfetta, and Stefano Facchini.
\newblock Quantum walking in curved spacetime: discrete metric.
\newblock {\em Quantum}, 2:84, 2018.

\bibitem{arrighi2019curved}
Pablo Arrighi, Giuseppe Di~Molfetta, Iv{\'a}n M{\'a}rquez-Mart{\'\i}n, and
  Armando P{\'e}rez.
\newblock From curved spacetime to spacetime-dependent local unitaries over the
  honeycomb and triangular quantum walks.
\newblock {\em Scientific reports}, 9(1):1--10, 2019.

\bibitem{arrighi2011unitarity}
Pablo Arrighi, Vincent Nesme, and Reinhard Werner.
\newblock Unitarity plus causality implies localizability.
\newblock {\em Journal of Computer and System Sciences}, 77(2):372--378, 2011.

\bibitem{berry2011two}
Scott~D Berry and Jingbo~B Wang.
\newblock Two-particle quantum walks: Entanglement and graph isomorphism
  testing.
\newblock {\em Physical Review A}, 83(4):042317, 2011.

\bibitem{brun2003quantum}
Todd~A Brun, Hilary~A Carteret, and Andris Ambainis.
\newblock Quantum to classical transition for random walks.
\newblock {\em Physical review letters}, 91(13):130602, 2003.

\bibitem{childs2009universal}
Andrew~M Childs.
\newblock Universal computation by quantum walk.
\newblock {\em Physical review letters}, 102(18):180501, 2009.

\bibitem{di2016discrete}
Giuseppe Di~Molfetta and Fabrice Debbasch.
\newblock Discrete-time quantum walks in random artificial gauge fields.
\newblock {\em Quantum Studies: Mathematics and Foundations}, 3(4):293--311,
  2016.

\bibitem{Dimolfetta2013aa}
Giuseppe Di~Molfetta, Fabrice Debbasch, and Marc Brachet.
\newblock Quantum walks as massless dirac fermions in curved space.
\newblock {\em Phys. Rev. A}, 88, 2013.

\bibitem{Dimolfetta2014aa}
Giuseppe Di~Molfetta, Fabrice Debbasch, and Marc Brachet.
\newblock Quantum walks in artificial electric and gravitational fields.
\newblock {\em Phys. A}, 397, 2014.

\bibitem{di2016quantum}
Giuseppe Di~Molfetta and Armando P{\'e}rez.
\newblock Quantum walks as simulators of neutrino oscillations in a vacuum and
  matter.
\newblock {\em New Journal of Physics}, 18(10):103038, 2016.

\bibitem{dong2010quantum}
Daoyi Dong and Ian~R Petersen.
\newblock Quantum control theory and applications: a survey.
\newblock {\em IET Control Theory \& Applications}, 4(12):2651--2671, 2010.

\bibitem{georgescu2014quantum}
Iulia~M Georgescu, Sahel Ashhab, and Franco Nori.
\newblock Quantum simulation.
\newblock {\em Reviews of Modern Physics}, 86(1):153, 2014.

\bibitem{grossing88}
Gerard Grossing and Anton Zeilinger.
\newblock Quantum cellular automata.
\newblock {\em Complex Systems}, 2:11, 1988.

\bibitem{Grover1996}
Lov~K. Grover.
\newblock A fast quantum mechanical algorithm for database search.
\newblock In {\em Proceedings of the Twenty-eighth Annual ACM Symposium on
  Theory of Computing}, STOC '96, pages 212--219, New York, NY, USA, 1996. ACM.

\bibitem{hatifi2019quantum}
Mohamed Hatifi, Giuseppe Di~Molfetta, Fabrice Debbasch, and Marc Brachet.
\newblock Quantum walk hydrodynamics.
\newblock {\em Scientific reports}, 9(1):1--7, 2019.

\bibitem{konno2010limit}
Norio Konno and Takuya Machida.
\newblock Limit theorems for quantum walks with memory.
\newblock {\em Quantum Information \& Computation}, 10(11):1004--1017, 2010.

\bibitem{li2016generic}
Dan Li, Michael Mc~Gettrick, Fei Gao, Jie Xu, and Qiao-Yan Wen.
\newblock Generic quantum walks with memory on regular graphs.
\newblock {\em Physical Review A}, 93(4):042323, 2016.

\bibitem{lovett2010universal}
Neil~B Lovett, Sally Cooper, Matthew Everitt, Matthew Trevers, and Viv Kendon.
\newblock Universal quantum computation using the discrete-time quantum walk.
\newblock {\em Physical Review A}, 81(4):042330, 2010.

\bibitem{mackay2002quantum}
Troy~D Mackay, Stephen~D Bartlett, Leigh~T Stephenson, and Barry~C Sanders.
\newblock Quantum walks in higher dimensions.
\newblock {\em Journal of Physics A: Mathematical and General}, 35(12):2745,
  2002.

\bibitem{marquez2017fermion}
Ivan M{\'a}rquez-Mart{\'\i}n, Giuseppe Di~Molfetta, and Armando P{\'e}rez.
\newblock Fermion confinement via quantum walks in (2+ 1)-dimensional and (3+
  1)-dimensional space-time.
\newblock {\em Physical Review A}, 95(4):042112, 2017.

\bibitem{mcgettrick2010one}
Michael McGettrick.
\newblock One dimensional quantum walks with memory.
\newblock {\em Quantum Information \& Computation}, 10(5):509--524, 2010.

\bibitem{Portugal2013}
Renato Portugal.
\newblock {\em Quantum Walks and Search Algorithms}.
\newblock Springer Publishing Company, Incorporated, 2013.

\bibitem{guillet2019grover}
Mathieu Roget, St{\'e}phane Guillet, Pablo Arrighi, and Giuseppe Di~Molfetta.
\newblock Grover search as a naturally occurring phenomenon.
\newblock {\em Physical Review Letters}, 124(18):180501, 2020.

\bibitem{rohde2013quantum}
Peter~P Rohde, Gavin~K Brennen, and Alexei Gilchrist.
\newblock Quantum walks with memory provided by recycled coins and a memory of
  the coin-flip history.
\newblock {\em Physical Review A}, 87(5):052302, 2013.

\bibitem{rohde2011multi}
Peter~P Rohde, Andreas Schreiber, Martin {\v{S}}tefa{\v{n}}{\'a}k, Igor Jex,
  and Christine Silberhorn.
\newblock Multi-walker discrete time quantum walks on arbitrary graphs, their
  properties and their photonic implementation.
\newblock {\em New Journal of Physics}, 13(1):013001, 2011.

\bibitem{shakeel2018quantum}
Asif Shakeel.
\newblock Quantum cellular automata: Schr$\backslash$"$\{$o$\}$ dinger and
  heisenberg pictures.
\newblock {\em arXiv preprint arXiv:1807.01192}, 2018.

\bibitem{shakeel2019neighborhood}
Asif Shakeel.
\newblock Neighborhood-history quantum walk.
\newblock {\em Physica Scripta}, 94(6):065207, apr 2019.

\bibitem{shakeel2014history}
Asif Shakeel, David~A Meyer, and Peter~J Love.
\newblock History dependent quantum random walks as quantum lattice gas
  automata.
\newblock {\em Journal of Mathematical Physics}, 55(12):122204, 2014.

\end{thebibliography}

\end{document}